\documentclass[twocolumn, nofootinbib, superscriptaddress, longbibliography]{revtex4-2}

\usepackage[pdftex]{graphicx}
\usepackage{mathrsfs,amsmath,amsthm}
\usepackage[colorlinks, breaklinks, urlcolor={blue}, linkcolor={red}, citecolor={blue}]{hyperref}
\usepackage{array}
\usepackage{amsmath}
\usepackage{type1cm}
\usepackage{lettrine}
\usepackage[english]{babel}
\usepackage{lmodern}
\usepackage{microtype}
\usepackage{booktabs}
\usepackage[T1]{fontenc}
\usepackage[boxed, vlined]{algorithm2e}
\usepackage{caption}
\usepackage{braket}
\usepackage{optidef}
\usepackage[nobreak=true]{mdframed}

\newmdtheoremenv{definition}{Definition}
\newmdtheoremenv{scheme}{Scheme}

\def\>{{\rangle}}
\def\<{{\langle}}
\def\depth{{\texttt{depth}}}

\newtheorem{theorem}{Theorem} 
\newtheorem{lemma}[theorem]{Lemma} 

\begin{document}

\title{Permutational-key quantum homomorphic encryption with homomorphic quantum error-correction}

\author{Yingkai Ouyang}
\email[]{y.ouyang@sheffield.ac.uk}
\homepage{http://www.quantumbespoke.com}
\affiliation{School of Mathematical and Physical Sciences, University of Sheffield, Sheffield, S3 7RH, United Kingdom}
\affiliation{Department of Engineering and Computer Science, National University of Singapore}
\affiliation{Centre of Quantum Technologies, National University of Singapore}

\author{Peter P. Rohde}
\email[]{dr.rohde@gmail.com}
\homepage{http://www.peterrohde.org}
\affiliation{Centre for Quantum Software \& Information (QSI), University of Technology Sydney, NSW 2007, Australia}
\affiliation{Hearne Institute for Theoretical Physics, Department of Physics \& Astronomy, Louisiana State University, Baton Rouge LA, United States}

\date{\today}

\frenchspacing

\begin{abstract}
The gold-standard for security in quantum cryptographic protocols is information-theoretic security. Information-theoretic security is surely future-proof, because it makes no assumptions on the hardness of any computational problems and relies only on the fundamental laws of quantum mechanics. Here, we revisit a permutational-key quantum homomorphic encryption protocol with information-theoretic security. We explain how to integrate this protocol with quantum error correction that has the error correction encoding as a homomorphism. This feature enables both client and server to apply the encoding and decoding step for the quantum error correction, without use of the encrypting permutation-key.
\end{abstract}

\maketitle

\section{Introduction}

Future quantum computing infrastructure is likely to be accessible to most end users via cloud-based services, owing to the high infrastructure cost. Similarly, the high-value of the applications under demand makes security considerations paramount \cite{bib:RohdeQI}. Quantum homomorphic encryption (QHE) is a class of cryptographic protocols allowing quantum computations to be performed on encrypted data without the server performing the computation having to first decrypt it, as is the case with conventional computational outsourcing. However, quantum computing differs from classical computing in that quantum error correction (QEC) is necessary to achieve fault-tolerance \cite{campbell2017roads} under realistic noise processes.

A salient feature of quantum cryptographic protocols is that they can be imbued with security based solely on the fundamental laws of quantum mechanics. In contrast with traditional cryptographic schemes whose security depends on the assumed hardness of certain computational problems, quantum protocols with information-theoretic (IT) security remain provably future-proof;  future developments in quantum computation or classical algorithms will never invalidate the security guarantee of cryptographic protocols with IT security.

Quantum cryptographic protocols are typically studied in ideal scenarios where no environmental noise is present. To combat environmental noise in these quantum protocols it is natural to consider their integration with quantum error correction (QEC) codes.

In the classical realm, a key attraction of homomorphic encryption (HE) \cite{Gen09} is that it serves as a cryptographic primitive from which a plethora of other cryptographic protocols can be derived \cite{alloghani2019systematic}. Given the potential of quantum generalizations of HE to usher in a plethora of quantum cryptographic protocols, QHE schemes have been studied under both information-theoretic security \cite{yu2014limitations, bib:ouyang2015quantum, tan2017practical, lai2018statistically,newman2017limitations, ouyang2020homomorphic} or with computational hardness assumptions \cite{BJe15, DSS16}. 

In this paper, we discuss how a permutational-key quantum homomorphic encryption scheme with IT security \cite{bib:ouyang2015quantum} can be infused with quantum error correction.
In such an error-corrected scheme, 
the encoding of the quantum error correction code 
is a homomorphism.
This fact allows both the client and the server to apply the encoding and decoding step for the quantum error correction, and quantum error correction is agnostic of the encrypting permutation-key.
Hence the server may perform all the costly quantum error correction steps, and the client need only perform the initial encoding of the quantum error correction code, and the final decoding of the quantum error correction code.

Since the choice of quantum error correction codes to be used with the scheme is highly flexible, we are able to discuss the fault-tolerant quantum error correction, and near-term error correction using permutation-invariant codes in the context of our error-corrected quantum homomorphic encryption 
scheme.

\section{Permutational-key quantum homomorphic encryption} \label{sec:pqhe}

\subsection{Algebra of quantum homomorphic encryption}

The client, Alice, prepares the quantum state $\rho$ which she would like to encrypt and subsequently outsource to the cloud for the execution of some computation. Before Alice sends her state to the server, she encrypts her state using the encryption process, ${\rm Encr}_\kappa(\cdot)$, where \mbox{$\kappa\in K$} denotes her chosen private key, chosen uniformly at random from a set of keys given by $K$. 
From Alice's perspective, who knows the key $\kappa$, her state is,
\begin{align}
    {\rm Encr}_\kappa(\rho)
    \label{eq:encrypted-client-perspective}.
\end{align}
The server, Bob, or indeed any eavesdropper, who does not know Alice's private key, instead perceives a state mixed over all possible keys Alice may have chosen,
\begin{align}
    {\rm Encr}(\rho ) 
    =\frac{1}{|K|} \sum_{\kappa \in K} {\rm Encr}_\kappa (\rho ).
    \label{eq:encrypted-server-perspective}
\end{align}

In this paper, we focus on the quantum homomorphic encryption schemes with information-theoretic security \cite{TKOCF-qhe,bib:ouyang2015quantum,ouyang2020homomorphic}.
This is in contrast to quantum homomorphic encryption schemes \cite{BJe15,alagic2017quantum} that
use classical fully homomorphic encryption \cite{DGHV10,GHS12} as a subroutine,
and thereby inherit the computational hardness assumptions of the constituent classical fully homomorphic encryption schemes.

Here, we quantify the security of Alice's data from the server using the maximal trace-distance between encrypted states, given by
\begin{align}
    \Delta = \max_{\rho, \rho' \in \mathscr S}
    \frac{1}{2}
    \| {\rm Encr}(\rho )  - {\rm Encr}(\rho' )   \|_1,\label{eq:QHE-security}
\end{align}
where $\| \cdot \|_1$ denotes the trace norm
and $\mathscr S$ is Alice's set of input density matrices. The closer $\Delta$ is to zero, the less distinguishable Alice's distinct inputs are to the server. Having $\Delta = 0$ would correspond to complete data-privacy.

Alice would like the server to compute on her encrypted quantum data. 
We denote the server's computation on what would be the unencrypted data as the quantum channel $\mathcal C$,
and the corresponding 
computation on the encrypted quantum data as the quantum channel $\overline{\mathcal C}$.
We denote $\mathscr A$ 
and $\overline {\mathscr A}$ as the set of all possible computations on unencrypted data and encrypted data, respectively.
For quantum homomorphic encryption, 
the computation $\overline{\mathcal C}$
is independent of Alice's secret key $\kappa$. 
In this paper, the server has no intention of hiding the desired computation from the server, 
and there is no circuit-privacy in the sense that $\mathcal C$ is public information.
Denote
\begin{align}
\overline {\mathscr S}
= \{ {\rm Encr} ( \rho ) : \rho \in \mathscr S\}
\end{align}
as the set of encrypted quantum states from the server's point of view.

From the client's perspective, the server's computation on the encrypted quantum state is
\begin{align}
  \overline{\mathcal C} \circ {\rm Encr}_\kappa (\rho ).
\end{align}
Since the server has no knowledge of $\kappa$, the channel $\overline{\mathcal C}$ must be independent of $\kappa$.
We assume that there exists a $\mathcal C_\kappa \in \mathscr A$ such that 
 \begin{align}
     \includegraphics[height=0.7cm]{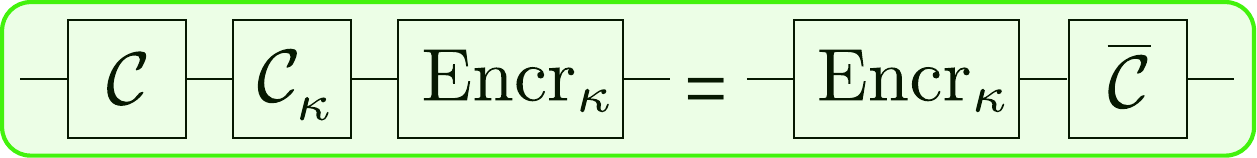}
     \label{Ic}.
 \end{align}
 In the special case where $\mathcal C_\kappa$ is the identity map, we obtain
 \begin{align}
     \includegraphics[height=0.7cm]{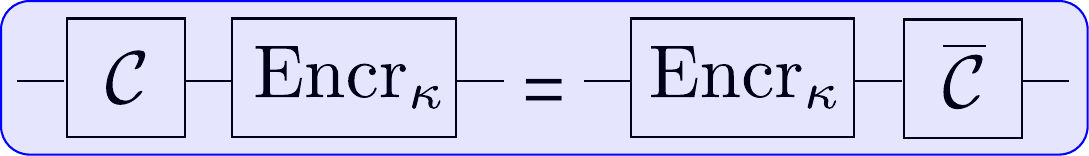}.
     \label{Ia}
 \end{align}

The scheme is termed as quantum homomorphic encryption, because a group homomorphism $\varphi$ relates $\overline{\mathcal C}$ to ${\mathcal C}$. Namely, there is a map $\varphi$ such that for every 
$\mathcal C_1, \mathcal C_2 \in \mathscr A$, the computations $\mathcal C_1$, $\mathcal C_2$ and 
$\mathcal C_1 \circ \mathcal C_2$ satisfy the algebraic relationship,
\begin{align}
    \varphi( \mathcal C_1 \circ \mathcal C_2 )
    = 
    \varphi( \mathcal C_1) \circ \varphi(\mathcal C_2 ).
\end{align}
Here, $\varphi( \mathcal C_1)$, $\varphi(\mathcal C_2 )$ and $\varphi( \mathcal C_1 \circ \mathcal C_2 )$ are what the server applies on $\overline{\mathscr S}$ to implement $\mathcal C_1$, $\mathcal C_2$ and $\mathcal C_1 \circ \mathcal C_2$ respectively.
For every $\mathcal C \in \mathscr A$, we write $\overline{\mathcal C} = \varphi(\mathcal C)$.
The group homomorphism $\varphi$ allows the server to compute any sequence $(\mathcal C_1, \dots, \mathcal C_s)$ using the sequence 
$(\varphi(\mathcal C_1), \dots, \varphi(\mathcal C_s))$.
We now summarize the notation of homomorphism in QHE.
\begin{definition}[Homomorphism in QHE]
Let $\mathscr A$ and $\overline{\mathscr A}$ be a group of unitary channels that acts on $\mathscr S$ and $\overline{\mathscr S}$ respectively.
The map $\varphi: \mathscr A \to \overline{\mathscr A}$ is a group homomorphism if for every $\mathcal C_1 , \mathcal C_2 \in \mathscr A$,
we have $\varphi(\mathcal C_1 \circ \mathcal C_2 ) = \varphi(\mathcal C_1) \circ \varphi(\mathcal C_2) $.
\end{definition}

The server upon completing the computation on the cipherspace $\overline{\mathscr S}$, returns the quantum state to the client. The client then decrypts this state with the decryption channel ${\rm Decr}_{\kappa,\mathcal C}(\cdot)$. 
The group homomorphism $\varphi$ that the server used before is only meaningful if after decryption, Alice recovers the correct state $\mathcal C(\rho)$. 
For the decryption to return the correct state, it suffices to require that 
 \begin{align}
     \includegraphics[height=0.7cm]{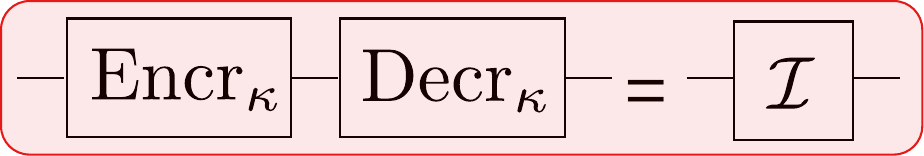}
     \label{Ib},
 \end{align}
where ${\rm Decr}_{\kappa} = {\rm Decr}_{\kappa,\mathcal I}$. 
Whenever ${\rm Encr}_{\kappa}$ is a unitary channel, the condition \eqref{Ib} is also equivalent to the requirement that ${\rm Decr}_{\kappa} = {\rm Encr}_{\kappa}^\dagger$.

The goal is that after decryption, Alice receives the correct state, no matter which secret key $\kappa$ she uses, and which computation $\mathcal C$ was evaluated. 
The following lemma tells us what decryption operation achieves this.
\begin{lemma}[Correctness]\label{lem:correctness}
Suppose that 
\begin{align}
     \includegraphics[height=0.55cm]{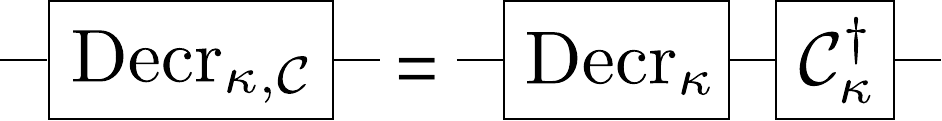}
     \label{If}.
 \end{align} 
 Then 
 \begin{align}
     \includegraphics[height=0.7cm]{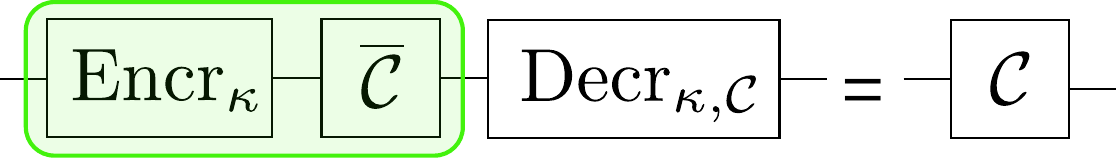}
     \label{Ig}
 \end{align}
for every $\kappa \in K$ and $\mathcal C \in \mathscr A$.
\end{lemma}
\begin{proof}
Using \eqref{Ib} and \eqref{Ic}, we find that
 \begin{align}
     \includegraphics[height=0.7cm]{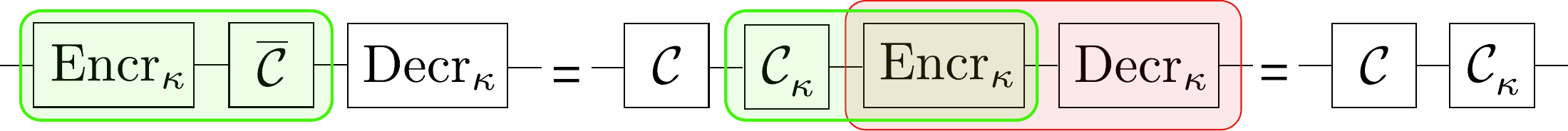}
     \label{Ie},
 \end{align}
from which the result can be deduced.
\end{proof}
The proof of Lemma \ref{lem:correctness} shows that the structure of decryption channel can be derived using the structure of the encryption channel and the delegated computation.

Apart from the homomorphism of $\varphi$ in a QHE scheme, the scheme also must keep the client's data secure, allow correct computation and compact decryption:
\begin{enumerate}

    \item {\bf Security:} 
    We say that the scheme is $\epsilon$-secure if 
    the maximal trace-distance, Eq.~(\ref{eq:QHE-security}), between encrypted states satisfies $\Delta \le \epsilon$.
    
    \item {\bf Correctness:} The QHE scheme is correct if for every key $\kappa \in K$ and every computation $\mathcal C \in \mathscr A$, there is an encryption channel ${\rm Encr}_\kappa$ and decryption channel ${\rm Decr}_{\kappa,\mathcal C}$ such that Eq.~(\ref{Ig}) holds.
    
    \item {\bf Compactness:} The QHE scheme is compact if the circuit complexity of the decryption channel is at most a polynomial in the key length.
\end{enumerate}
Here, {\em compactness} is a condition that sometimes ensures that Alice's decryption ${\rm Decr}_{\kappa,\mathcal C}(\cdot)$ is much easier to perform than the actual computation $\mathcal C$. 
Of course, the circuit complexity is dependent on the gate-set used for quantum computation. 
In most cases, the Boykin gate set \cite{boykin}, which comprises of Clifford gates and T-gates is used.
In the most extreme case when ${\rm Decr}_{\kappa,\mathcal C} = {\rm Decr}_{\kappa}$,
which is the case for the permutational key quantum homomorphic encryption scheme that we later focus on, the decryption is independent of $\mathcal C$ and the condition of compactness depends only on the properties of the encrypting key $\kappa$.






\section{Permutational-key quantum homomorphic encryption}\label{sec:apps}


\subsection{Revisiting quantum homomorphic encryption scheme from quantum codes}

Here, we review the quantum homomorphic encryption scheme from quantum codes as given in Ref.~\cite{bib:ouyang2015quantum}.

In the first step, the client encodes each qubit of the quantum data into a random quantum error correction code, where the code is not used for quantum error correction,
but rather, for mapping each qubit,
$   2^{-1}\sum_{j=0}^3 a_j \sigma_j,$
to 
$    2^{-m}\sum_{j=0}^3 a_j \sigma_j^{\otimes m},$
where $\sigma_0 = I, \sigma_1 = X, \sigma_2 = Y, \sigma_3 = Z$
are the Pauli matrices.
We denote such as a mapping as
\begin{align}
    {\rm RandCode}( 2^{-1}\sum_{j=0}^3 a_j \sigma_j)
    =2^{-m}\sum_{j=0}^3 a_j \sigma_j^{\otimes m}.
\end{align}
We can achieve the encoding ${\rm RandCode}$ by 
preparing $m-1$ maximally mixed states for each data qubit,
and applying $2m-2$ CNOT gates \cite{bib:ouyang2015quantum}.
This is essentially a computation of depth 1,
and we can implement this 
using a collective rotations on single qubits, and Ising interactions that follow a star graph topology.

In the second step, one introduces $m$ additional maximally mixed states for each $m$ qubit random code, and permutes the order of these $2m$ qubits according to a permutational key $\kappa$, where
$\kappa$ is labelled by elements of $S_{2m}$,
the symmetric group of order $2m$.
Each data qubit is thus encrypted into $2m$ qubits using the same permutational key $\kappa$.
The encrypted qubits are then sent to the server. 
When the initial number of data qubits is $r$,
the number of qubits sent to the server is $2mr$.

When $m$ is an odd number, the server performs a homomorphic evaluation of Clifford gates on the encrypted quantum data by applying identical Clifford gates on each of the $2m$ qubits. 
This is because the action of the $m$ identical Clifford gates on sets of qubits that do not hold quantum data is trivial. 
In particular, the server implements single-qubit logical Clifford gates on the encrypted data by applying transversal Clifford gates on of $2m$ qubits.
The server implements logical CNOTs by applying transversal CNOTS similarly on $2m$ qubits.

Each non-Clifford $T$-gate can be implemented with a heralded 50 percent probability of success with the help of a logical gate teleportation protocol.
The logical gate teleportation protocol consumes an encrypted magic state, 
entangles the encrypted magic state with the data via transversal CNOTs, and measures the encrypted magic state. Success of the logical gate teleportation is heralded by the client based on the parity of the measurement outcome of qubits in the computational basis.

When the client encrypts $r$ data qubits, 
the trace distance between the encrypted quantum data from the server's point of view and the maximally mixed state is at most,
\begin{align}
    \Delta(r,m) = \sqrt{2^r/|K|} = 2^{r/2}\binom{2m}{m}^{-1/2}.
    \label{eq:permutation-key-security}
\end{align} 
This means that the data-privacy of this quantum homomorphic encryption scheme is $\Delta = 2\Delta(r,m).$
Hence, the scheme's data privacy approaches 0 exponentially fast in $m$ \cite{bib:ouyang2015quantum} for sufficiently slow growing $r$.
Note that the number $r$ must include the number of magic states that the client encodes.
So if the initial size of the quantum data is $s$ qubits, then the number of available magic states is $r-s$ when $r \ge s$.

The success probability of implementing $r-s$ $T$ gates can be boosted to 100 percent if we are willing to allow the client to interact with the server, and use additional ancilla qubits. In such a situation, however, the scheme will no longer be a quantum homomorphic encryption scheme, but rather, be a form of secure delegated quantum computation. 
For every encrypted magic state, the client will need to prepare an encrypted $|0\>$ state and an encrypted $|1\>$ state, and has these encrypted states assigned in random rows.
Using these three encrypted ancilla states, the client can help the server perform deterministic $T$-gates. 

For this, the server first transmits classical bits containing the measurement outcomes on the encrypted magic state to the client. Second, the client uses knowledge of the secret permutation key to evaluate the parity of a subset of these classical bits that correspond to the secret permutation. Third, based on this parity, the client tells the server the which row, where an encrypted $|0\>$ or encrypted $|1\>$ state resides, should be used to perform a controlled correction for the gate-teleportation process. 
Without the permutation key, the classical label is completely independent of the quantum data \cite{ouyang2017computing} and of the private key. 
Hence in spite of the server's classical communication with Alice, the server learns no additional information about the quantum data. 
With this method, the server can perform any number of $T$-gates, provided that Alice prepares the encrypted magic states on demand. The only cavaet is that the security of the scheme worsens exponentially with the number of encrypted $T$-gates consumed (see Eq.~(\ref{eq:permutation-key-security})).


\subsection{Quantum error correction as a homomorphism}

Here, we will understand how the encoding map that introduces quantum error correction capabilities to the above permutational-key 
quantum homomorphic encryption scheme 
is a homomorphism.
We begin by explaining how our proposed quantum error correction encoding 
integrates with 
permutational-key quantum homomorphic encryption. 

Recall that the encryption map 
${\rm Encr}_{\kappa}$
of the permutational-key quantum homomorphic encryption scheme 
takes each data qubit into $2m$
qubits. 
It is important to note that for every data qubit $\rho$,
we can write 
\begin{align}
    {\rm Encr}_{\kappa}(\rho) = 
    \pi_\kappa
    \left( {\rm RandCode}(\rho) \otimes (I/2)^{\otimes m} \right)
    \pi_\kappa ^\dagger,\label{qhe:encryption}
\end{align}
where $\pi_\kappa$ is a representation of $\kappa$ that permutations $2m$ qubits according to the permutation $\kappa \in S_{2m}$.
Also, the encryption map acts identically on each data qubit the client has. 

Now, let us consider an encoding of a qubit into a quantum error correction code, given by the quantum channel $\mathcal Q$.
Namely, the channel 
$\mathcal Q$ maps a single qubit into an $n$ qubit quantum error correction code with distance $d$,
and hence can also be specified as a
$[[n,1,d]]$ code.
We let $\mathcal Q$ act identically on each of the $2m$ qubits in 
the right side of \eqref{qhe:encryption}, 
which then gives the state
$\mathcal Q^{\otimes 2m} (
    {\rm Encr}_{\kappa}(\rho) 
) $
for every unencrypted data qubit $\rho$.

Since the channel $\mathcal Q^{\otimes 2m}$ is an i.i.d channel and invariant under permutations, we have
\begin{align}
&\mathcal Q^{\otimes 2m} (
{\rm Encr}_{\kappa}(\rho)
)\notag\\
=&
\pi_\kappa
\left(
\mathcal Q^{\otimes 2m}
\left( {\rm RandCode}(\rho) \otimes \mathcal (I/2)^{\otimes m} \right)\right)
\pi_\kappa ^\dagger,\label{qhe:QEC-homo}
\end{align}
which amounts to having the QEC encoding as a homomorphism.
Thus, applying $\mathcal{Q}^{\otimes 2m}$ to ${\rm RandCode}(\rho) \otimes (I/2)^{\otimes m}$ is equivalent to applying it to the encrypted state ${\rm Encr}_\kappa(\rho)$ after permutation $\pi_\kappa$. Consequently, the server can perform QEC on encrypted states received from the QHE scheme, requiring only prior agreement on the QEC code (e.g., via public symmetric key distribution). Keeping the QEC code secret from adversaries further enhances security by randomizing the state over all possible codes.

After the server completes the computation, 
the server may return all the qubits back to the client.
Then the client can decode the quantum code, and subsequently decrypt the quantum information.
Alternatively, for a client with less ability to perform quantum computation, 
the client can have the server perform logical measurements on all of the encoded qubits, and send the classical information about these measurement outcomes to the client, after which the 
client processes of the classical information in the decryption step.

\subsubsection{Infusing fault-tolerance}

Fault-tolerant quantum computation will be possible if we choose $\mathcal Q$ to be any fault-tolerant QEC code.
In the paradigm of fault-tolerant quantum computation, we accept the inevitability that each elementary gate fails with some probability. In spite of noise, carefully designed fault-tolerant quantum circuitry can perform reliable quantum computation. Quantum information will be encoded into a fault-tolerant family of quantum error correction codes $\Omega =\{ \mathcal Q_t : t \ge 1 \}$ that corrects an increasing number of errors $t$.

Suppose that each physical gate fails independently with probability $p_0$, and we use a quantum code $\mathcal Q_t$ that corrects $t$ errors. We say that $\Omega$ has a fault-tolerant threshold $p_{\Omega}$ if for every positive integer $t$, the logical failure probability $p_{t,\Omega}$ of each fault-tolerant gate in $\mathcal Q_t$ satisfies the inequality,
\begin{align}
    p_{t,\Omega} \le a_\Omega (p_0 / p_{\Omega})^t,
\end{align}
for some positive $a_\Omega$. 

Fault-tolerant syndrome extraction allow us to reliably extract information about errors as they occur. Compared to conventional non-fault-tolerant syndrome extraction, fault-tolerant syndrome extraction has additional overhead in terms of the number of ancillary qubits required. Recent advances in flag-fault-tolerant quantum computation \cite{chao2018flag, Chao2020.anyflag} allows the number of these ancilla to be bounded for each round of fault-tolerant syndrome extraction. In the Aliferis, Gottesman and Preskill framework of fault-tolerant quantum computation \cite{AGP05}, each fault-tolerant gadget is both preceded and followed by rounds of fault-tolerant syndrome extraction. Hence, for the fault-tolerant evaluation of a quantum circuit of depth $\depth$ on $r$ qubits, the total number of fault-tolerant syndrome extractions required for fault-tolerant evaluation is at most $r(\depth + 1)$. Each fault-tolerant syndrome extraction requires at most
\begin{align}
    A_{n,t} = 2tn[t(t+2)+n] \binom t 2
    \le t^3 n (t^2+2t+n),
\end{align}
ancillary qubits \cite{Chao2020.anyflag}\footnote{This equation does not directly appear in this reference, but through private communication with Chao was shown to follow from these results.} for a length $n$ quantum error correction code that corrects $t$ errors.
The server can fully prepare and manipulate these ancillary qubits, so the client does not need to prepare or send them.


\subsubsection{Permutation-invariant code implementation}

If the physical qubits of the server encounter errors such as untracked qubit loss (deletions), or untracked qubit insertions, such errors cannot be straightforwardly corrected using traditional quantum error correction codes.
Fortunately, 
 permutation-invariant codes \cite{Rus00,PoR04,ouyang2014permutation,ouyang2015permutation,OUYANG201743,movassagh2020constructing,ouyang2024measurement}
 can correct deletions \cite{ouyang2021permutation,ouyang2022finite,hagiwara2021four} and insertion errors 
\cite{hagiwara2021four,shibayama2021equivalence,shibayama2022equivalence}.

Permutation-invariant codes are a special family of QEC codes that are invariant under any permutation of the underlying particles, and it is this symmetry, on top of their distance, that 
allows the correction of insertion and deletions.

Permutation invariant codes, apart from being able to correct insertions and deletions, 
also admit near-term implementations. 
This is because the controllability of permutation-invariant codes by global fields, which could allow for their scalable physical implementations \cite{johnsson2020geometric} in near-term devices such as trapped ions or ultracold atoms where addressability is challenging due to cross-talk. 

To allow fault-tolerance, permutation-invariant codes can be used as inner codes, in a code-concatenation with outer codes that allow fault-tolerant gates, such as certain quantum LDPC codes. 
Initial steps in this direction has been made, for example in the study of the belief-propagation type decoding of erasure errors in permutation-invariant codes concatenated with quantum LDPC codes \cite{kuo2024degenerate}.

\subsection{Correctness, data-security and compactness of error-corrected scheme}

{\bf Correctness:-}
The correctness of the error-corrected scheme is guaranteed as long as there are not too many errors that afflict the scheme.
The only way for an honest server to deviate from the correctness condition is when the error-correction step experiences a logical error. 
Provided that the error rates in the scheme are below the fault-tolerant threshold, the correctness of the error corrected protocol can be made arbitrarily good, at the expense of introducing an increasing amount of qubit overhead.

{\bf Data-security:-}
The security of the permutation-key scheme worsens exponentially with the number $r$ of $T$-gates and ancilla qubits used. However, the security of the permutation-key scheme improves exponentially with $m$. Hence if we increase $m$ together with $r$, we can guarantee a constant level of security. 
This data-security of the error-corrected scheme is identical to that of the vanilla QHE scheme for an honest server. 
Conditioned on the correctness of the 
error-corrected protocol, 
the if the server is dishonest, 
the 
data-security of the error-corrected scheme is identical to that of the vanilla scheme. Intuitively, this is because of the fundamental information-disturbance tradeoff relation, which can be applied in our scheme using Ref.~\cite{ouyang2017computing,ouyang-approx}.

{\bf Compactness:-}
The compactness of the error-corrected protocol is identical to the compactness of the vanilla QEC protocol.

First we consider the 
case where after the server completes the computation, 
the server sends all the qubits back to the client, and the client has to decode the quantum error correction code.
 The decryption algorithm involves
 two steps (1) first decoding the quantum error correction code, and second (2) unpermuting the columns of qubits. 
For step (1), this depends on the complexity of the decoding circuit of the quantum error correction code. We may choose $n$ to be linear in $m$, and choose quantum error correction codes with efficient decoding circuits, so that the complexity of the decoding circuits is at most polynomial in $m$.

 Step (2) involves unpermuting $2m$ columns of qubits arranged in $rn$ rows, where $r$ is the data qubits the client uses. Hence, the number of swaps needed to perform the permutation for the decryption is $O(rnm)$. When $m$ is linear in $r$ (which allows the data-security condition of the scheme to hold), the complexity of decryption is $O(r^2 {\rm poly}(r))$,
 and in such a scenario, the scheme is compact. 

When the server sends just classical information about the logical measurements back to the client, the decryption algorithm can have some savings. All that is needed then is to identify $m$ out of $2m$ columns that correspond to the $rm$ bits of data needed for the decryption step.  The identification procedure only needs $O(rm)$ steps, and computing the correct parity of the measurement outcomes as per Ref.~\cite{ouyang2017computing} involves $O(m)$ steps, and when $m$ is linear in $r$,
this means that the decryption step is $O(r^2)$ in complexity, and the scheme is also compact.

 \section{Conclusion} \label{sec:conclusion}

We have discussed how a permutational-key quantum homomorphic encryption scheme can be imbued with a form of homomorphic quantum error correction.
The homomorphism allows both the client and the sever to perform quantum error correction, though the largest benefit is having the server to be able to perform all the difficult quantum error correction steps without compromising the security of the vanilla scheme.




\section*{Acknowledgements}
We thank Rui Chao for discussions, and thank Yanglin Hu and Marco Tomamichel for their feedback.
Peter Rohde is funded by an ARC Future Fellowship (project FT160100397). This research was conducted by the Australian Research Council Centre of Excellence for Engineered Quantum Systems (project CE170100009) and funded by the Australian Government. Yingkai Ouyang is supported by the Quantum Engineering Programme grant NRF2021-QEP2-01-P06, and also in part by NUS startup grants (R-263-000-E32-133 and R-263-000-E32-731), and the National Research Foundation, Prime Minister’s Office, Singapore and the Ministry of Education, Singapore under the Research Centres of Excellence program.
Y.O. also acknowledges support from EPSRC Grant No. EP/W028115/1 and also the EPSRC funded QCI3 Hub under Grant No. EP/Z53318X/1.

\bibliography{qecqhe}

\end{document}